\documentclass[a4paper, 11pt]{article}

\usepackage{caption, titlesec}
\usepackage[english]{babel}

\usepackage{hyperref}
\titleformat{\section}{\large\bfseries\filcenter}{\thesection}{1em}{}
\titleformat{\subsection}{\bfseries}{\thesubsection}{1em}{}
\titleformat{\subsubsection}[runin]{\bfseries}{\thesubsubsection}{1em}{}[.]
\makeatletter
\newcommand*{\rom}[1]{ #1}
\makeatother
 \usepackage[miktex]{gnuplottex}
 \usepackage{tikz}

\usetikzlibrary{positioning}
\usetikzlibrary{decorations.text}
\usetikzlibrary{decorations.pathmorphing}

\usetikzlibrary{arrows} 
\tikzset{
    >=stealth',
    pil/.style={
           ->,
           thick,
           shorten <=2pt,
           shorten >=2pt,}
}

\usepackage{setspace}
\usepackage{amsmath, amssymb, amsthm, mathrsfs}

\usepackage{graphicx}
\usepackage{pstricks}

\usepackage[backend=bibtex]{biblatex}
\bibliography{egl3.bib}
\usepackage[left=3cm,right=3cm, top=2.5cm,bottom=2.5cm,bindingoffset=0cm]{geometry}

\theoremstyle{plain}
\newtheorem{proposition}{Proposition}[section]
\newtheorem{theorem}[proposition]{Theorem}

\theoremstyle{definition}
\newtheorem{definition}[proposition]{Definition}
\newtheorem{remark}[proposition]{Remark}
\newtheorem{example}[proposition]{Example}

\newcommand{\gl}{\mathfrak{gl}}
\newcommand{\GL}{\mathrm{GL}}
\newcommand{\so}{\mathfrak{so}}

\newcommand{\R}{\mathbb{R}}
\newcommand{\tr}{\mathrm{tr}\,}

\newcommand{\Jac}{\mathrm{Jac}}

\newcommand{\Res}{\mathrm{Res}}
\newcommand{\diff}[1]{{d}  #1}
\newcommand{\diffXp}[1]{ \frac{\diff }{\diff #1} }
\newcommand{\diffFXp}[2]{ \frac{\diff #1}{\diff #2} }
\newcommand{\Complex}{\mathbb{C}}
\newcommand{\Id}{\mathrm{Id}}
\newcommand{\rank}[1]{\mathrm{rank} \, #1}

\newcommand{\CP}{\mathbb{CP}}

\newcommand{\RP}{\mathbb{RP}}
\newcommand{\g}{\mathfrak{g}}

\usepackage{enumitem}
\newlist{longenum}{enumerate}{5}
\setlist[longenum,1]{label=\roman*)}
\setlist[longenum,2]{label=\alph*)}

\numberwithin{equation}{section}

\title{Euler equations on the general linear group, cubic curves, and inscribed hexagons}
\author{Konstantin Aleshkin\thanks{SISSA and Landau Institute for Theoretical Physics, e-mail: \tt{kaleshkin@sissa.it}}
\, and Anton Izosimov\thanks{University of Toronto, e-mail: \tt{izosimov@math.utoronto.ca}}}
\date{}
\begin{document}
\maketitle
\begin{abstract}
We study integrable Euler equations on the Lie algebra $\gl(3,\R)$ by interpreting them as evolutions on the space of hexagons inscribed in a real cubic curve.

\end{abstract}

\section{Introduction}

In this paper we study an integrable matrix differential equation
\begin{align}
\label{mainEquation}
\diffXp{t}{X} = [X^2, A]
\end{align}

where $X \in \gl(3,\R)$ is a real $3$ by $3$ matrix depending on the time $t$, and $A \in \gl(3,\R)$ is a real, fixed, $3$ by $3$ matrix.  While we will not need a precise definition of an integrable system, we will take the point of view of \cite{hitchin2013integrable} according to which \textit{ ``integrability of a system of differential equations should manifest itself through some generally recognizable features: i) the existence of many conserved quantities, ii) the presence of algebraic geometry, iii) the ability to give explicit solutions.''} The above equation shows all these properties. In particular, the algebraic geometry underlying this equation is the geometry of real cubic curves. {This algebraic geometry arises from the possibility to rewrite equation \eqref{mainEquation}  in the so-called \textit{Lax form with a spectral parameter} (see equation \eqref{LaxFlow} below). Lax representation with spectral parameter for equations of type \eqref{mainEquation} was found in the fundamental S.V.\,Manakov's paper \cite{Man}. }

\par

Equation  \eqref{mainEquation}  can be regarded as a special case of several general constructions of integrable systems. In particular, it can be obtained by the argument shift method{~\cite{Man, MF}}, or by the method based on loop algebras \cite{Adler, Adler2, Reyman1, Reyman2}. Depending on the restrictions imposed on the matrices $X$ and $A$, this equation is known under different names. If $A$ is symmetric, and $X$ is skew-symmetric, it becomes the classical Euler equation describing the rotation of a rigid body with a fixed point. If, on the contrary, $A$ is skew-symmetric, and $X$ is symmetric, this equation is known as the Bloch-Iserles system~\cite{bloch2006isospectral, bloch2009class}. Finally, it is worth mentioning the case of skew-Hermitian $X$ and Hermitian $A$. In this setting, the above equation {describes travelling wave solutions} for the famous three-wave equation arising in optics, fluid dynamics, and plasma physics 
\cite{Alb98}. 

\par

In the present paper, we regard equation \eqref{mainEquation} as a dynamical system on the whole space $\gl(3,\R)$ of real $3$ by $3$ matrices. A distinctive feature of this full  system is that, in contrast to the symmetric and skew-symmetric cases, solutions on arbitrary matrices need not be bounded and, in particular, may blow up in finite time. Using algebro-geometric technique, we show that for a generic matrix $A$ the above equation has both types of solutions, that is blow-up solutions, and solutions defined for any $t \in \R$. We also show that the behavior of a given solution can be understood in terms of a simple geometric construction. 
Namely, with each generic initial condition $X$ we associate a real cubic curve $C_X$ {with fixed points at infinity (the \textit{spectral curve} coming from the Lax representation)}, and a hexagon $H_X$ inscribed in this curve in such a way that its sides are parallel to the asymptotes of the curve. Then, we show that the behavior of the solution of  the above equation with initial condition $X$ is completely determined by the number of ovals of the curve $C_X$ and the distribution  of vertices of the hexagon $H_X$ among these ovals. 
\par\smallskip
Among other possible interpretations, the above equation may be regarded as an Euler equation on the Lie algebra $\gl(3,\R)$, or, which is the same, the geodesic flow of a certain left-invariant metric on the general linear group $\GL(3,\R)$. The study of such metrics originates from V.\,Arnold's fundamental paper~\cite{Arn66}, where Arnold suggests a common geometric framework for the Euler equation governing the motion of an ideal fluid, and the Euler equation in rigid body dynamics. In Arnold's approach, both equations describe the geodesic flow of a one-sided invariant metric on a certain Lie group $G$. Such a geodesic flow is a dynamical system on the cotangent bundle ${T}^*G$, and, thanks to the $G$-invariance, it descends to the quotient space ${T}^*G \,/\, G$, which is naturally identified with the dual Lie algebra~$\g^*$. The corresponding equation on $\g^*$ is called an Euler equation.\par Equation \eqref{mainEquation} is an example of an Euler equation on the Lie algebra $\gl(3,\R)$.  It describes the geodesic flow of a left-invariant pseudo-Riemannian metric $(\,, )$ on the group $\GL(3,\R)$ given at the identity by
$$( X, X ) := \tr  X \mathcal{A}^{-1}(X)$$
where $\mathcal A(X) = \frac{1}{2}(AX + XA)$. In particular, the problem of existence of global solutions for equation \eqref{mainEquation} is equivalent to the problem of geodesic completeness for the metric~$(\,, )$. 
\par
A distinctive feature of the Euler equation  \eqref{mainEquation} is its integrability. 
Note that a general Euler equation need not be integrable, and integrable examples are in fact quite rare. In particular, the above equation seems to be the only known example of an integrable Euler equation on $\gl(3,\R)$. 
\par
The problem of geodesic completeness for left-invariant metrics on finite-dimensional Lie groups was studied, for example, in \cite{alekseevskii1987completeness, bromberg1998completude}. Note that for general, non-integrable, metrics, geodesic completeness or, equivalently, existence of global in time solutions of the Euler equation, seems to be a very difficult problem. \par
The classification problem for global and blow-up solutions of equation \eqref{mainEquation}, as well as the corresponding problem for other integrable systems constructed by the argument shift method, was considered  for the first time  in \cite{BIKO}. Note that this problem is of particular interest in the integrable case in connection with the Arnold-Liouville theorem \cite{Arn89}. Recall that this theorem asserts that the phase space of a completely integrable system is almost everywhere foliated into invariant tori, or, in the non-compact case, into invariant cylinders. However, this theorem does not apply for those fibers of the system which contain blow-up solutions.

\par
We also remark that since equation \eqref{mainEquation} is integrable, its solutions can be explicitly expressed in terms of theta functions. So, global behavior of solutions can be, in principle, studied by finding and examining explicit formulas. 
However, as we show in the present paper, global properties of solutions can be in fact understood from purely geometrical considerations, and there is no need in the analysis of complicated theta-functional formulas.  \par
We tried to make the exposition self-contained. In particular, we do not assume that the reader is familiar with the general theory of integrable systems and the algebro-geometric approach to such systems. For most statements which can be, in principle, derived from this general theory, we give geometric proofs (relations to the general theory are explained in remarks; see, in particular, Remark \ref{relGen}). The only exception is, perhaps, Proposition~\ref{velocityFormula} where we follow the standard approach on linearization of an integrable flow on the Jacobian. It would be interesting to find a geometric proof for this statement as well.
\par
Main results of the paper are in Section \ref{results}. Section \ref{proofs} is devoted to proofs of these results. In Section \ref{discussion}, we discuss possible generalizations of our approach to the $\gl(n)$ case and their relation to general questions of real algebraic geometry. 

\par
\smallskip
{\bf Acknowledgments.}
The second author was partially supported by the Dynasty Foundation Scholarship and an NSERC research grant. The authors are grateful to Alexey Bolsinov and Boris Khesin for useful remarks.

\section{Main constructions and results}\label{results}

\subsection{Reduction to diagonal matrices}
In what follows, we assume that the eigenvalues of the matrix $A$ are all distinct and real. {The case of complex conjugate eigenvalues can be treated using similar ideas, but still needs a separate consideration, and we omit it}.\par
 Denote the eigenvalues of $A$ by $a_1$, $a_2$, $a_3$. Note that equation \nolinebreak \eqref{mainEquation} is invariant under similarity transformations
$$
X \mapsto BXB^{-1}, \quad A \mapsto BAB^{-1}.
$$
For this reason, we may assume that $A$ is a diagonal matrix with diagonal entries $a_1$, $a_2$, $a_3$. Therefore, $9$-dimensional family of equations \eqref{mainEquation} boils down to a $3$-dimensional family parametrized by three real numbers $a_1$, $a_2$, $a_3$. \par

So, in what follows, we always assume that $A$ is diagonal with distinct diagonal entries. We call such diagonal matrices \textit{generic}.\par
Provided that $A$ is diagonal, equation~\eqref{mainEquation} is invariant under transformations of the form $X \mapsto DXD^{-1}$ where $D$ is an invertible {diagonal} matrix. Such transformations form a group which may be regarded as the quotient group of invertible diagonal matrices by scalar matrices. We shall denote this quotient group by $\mathbb P\mathrm D(3,\R)$. This group is isomorphic to $(\R^*)^2$, and in particular, it is disconnected. We denote its connected component of the identity by $\mathbb P\mathrm D^+(3,\R)$. The latter group consists of (cosets of) those diagonal matrices whose diagonal entries are of the same sign.

\subsection{Lax representation and spectral curve}\label{LaxRepr}

We begin our study of equation \eqref{mainEquation} by rewriting it as a so-called \textit{Lax equation}

\begin{align}\label{LaxFlow}
\diffXp{t}{X_\lambda}= [X_\lambda, Y_\lambda]
\end{align}
where 
\begin{align}\label{xy}
X_\lambda := X + \lambda A, \quad Y_\lambda := AX + XA + \lambda A^2\end{align} 
and $\lambda \in \Complex$ is an auxiliary time-independent parameter, called the \textit{spectral parameter}. It is straightforward to verify that equations \eqref{mainEquation} and  \eqref{LaxFlow}  are equivalent. 
{\begin{remark}For details about Lax equations with a spectral parameter and their algebraic-geometric solutions see, e.g., the monograph \cite{babelon}.\end{remark}}

The following proposition is well-known.
\begin{proposition}\label{timeInd} If a matrix $X_\lambda$ evolves according to equation \eqref{LaxFlow}, then the eigenvalues of $X_\lambda$ do not change with time.
\end{proposition}
\begin{proof}
Using induction on $k$, one can show that
$$
\diffXp{t}{X_\lambda^k}= [X_\lambda^k, Y_\lambda]
$$
 for any integer $k \geq 1$; therefore
$$
\diffXp{t}{\tr X_\lambda^k}= \tr [X_\lambda^k, Y_\lambda] = 0
$$
where in the last identity we used that the trace of a commutator is always equal to zero. Thus, since traces of powers of $X_\lambda$ do not depend on $t$,  neither do its eigenvalues, q.e.d.
\end{proof}
Proposition \ref{timeInd} implies that the coefficients of the characteristic polynomial $f_X(\lambda, \mu) := \det(X+\lambda A - \mu \Id)$ are conserved along the solutions of equation \eqref{mainEquation}. Note that only six out of ten coefficients explicitly depend on $X$, so there are six conserved quantities. 
We will not need explicit expressions for these conserved quantities. Instead, we organize them into an algebraic curve, called the {spectral curve}. In affine coordinates, this curve is defined by the equation $f_X(\lambda , \mu) = 0$. However, it will be convenient for us to work in homogenous coordinates. For this reason, we give the following definition:

\begin{definition} For a given $X \in \gl(3,\R)$, the curve
$$
 C_{X}:= \{ (z_1 : z_2 :z_3) \in \CP^2 \mid \det(z_3 X +  z_1 A - z_2 \Id) = 0 \}
$$
is called the \textit{spectral curve}.
\end{definition}
By definition, the spectral curve $C_X$ is conserved along solutions of equation \eqref{mainEquation}. 
\begin{proposition}\label{scProp} The spectral curve $C_X$ is a real\footnote{Recall that an algebraic curve is called real if it is invariant under complex conjugation, or, equivalently, if it can be defined as the zero locus of a real polynomial. } projective cubic intersecting the line at infinity $\{z_3 = 0\}$ at points $$\infty_1 = (1: a_1: 0), \quad \infty_2 = (1: a_2: 0), \quad \infty_3 = (1: a_3: 0)$$
where $a_1, a_2, a_3$ are the eigenvalues of $A$.
\end{proposition}
\begin{proof} The proof is straightforward. \end{proof}
As we show below, any smooth real cubic curve passing through the points $\infty_1$, $\infty_2$, $\infty_3$ is the spectral curve for a suitable matrix $X \in \gl(3,\R)$. Moreover, we explicitly describe the topology of the set of matrices $X$ corresponding to the given curve $C$ in terms of the geometry of $C$.

\subsection{Isospectral sets}
Since the flow \eqref{mainEquation} preserves the spectral curve associated with the matrix $X$, we can restrict this flow to the set 
$$
  \mathcal T_C = \{ X \in \gl(3,\R) : C_X = C\}
$$
of matrices whose spectral curve is the same. Note that the set $\mathcal T_C$ may also be defined as a joint level set for six conserved quantities of equation \eqref{mainEquation}  (recall that these conserved quantities are, by definition, the coefficients of the equation of the spectral curve).
Since the space $\gl(3, \R)$ is $9$-dimensional, we should expect that the set $\mathcal T_C$ is generically of dimension $9 - 6 = 3$.\par
Further, note that the flow \eqref{mainEquation} restricted to the $3$-dimensional manifold $\mathcal T_C$ has a $2$-dimensional symmetry group $\mathbb P\mathrm D^+(3,\R)$ acting by conjugation. A dimension count suggests that the $\mathbb P\mathrm D^+(3,\R)$ orbits of solutions of equation~\eqref{mainEquation} are exactly the connected components of $\mathcal T_C$. In particular, all solutions lying in the same connected component should have the same global behavior. \par 
 In what follows, we aim to answer the following questions.
\begin{longenum}
\item For a given a cubic curve $C$, what is the topology of the set $\mathcal T_C$? In particular, how many connected components does it have?
\item What is the global behavior of solutions of \eqref{mainEquation} on each of these components? In particular, do these solutions blow up or exist for all times?
\item Given an initial condition $X \in \gl(3,\R)$, how do we determine whether the solution passing through $X$ blows up, or exists for all times?
\end{longenum}
The answer to the first two of these questions is given by Theorem \ref{surfacetheorem}. The answer to the third question is given by Theorem \ref{completenesstheorem}.

 \subsection{Topology of isospectral sets}
 \begin{figure}[t]
 \centering
\begin{tabular}{ccc}
\multicolumn{1}{l}{(a)} & \multicolumn{1}{l}{(b)} & \multicolumn{1}{l}{(c)} \\
%\\

\includegraphics[scale = 1]{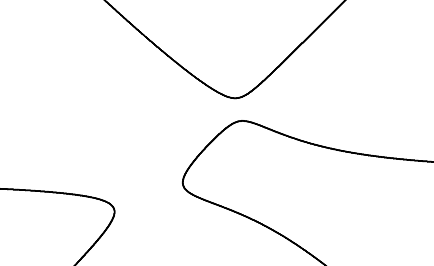} & \includegraphics[scale = 1]{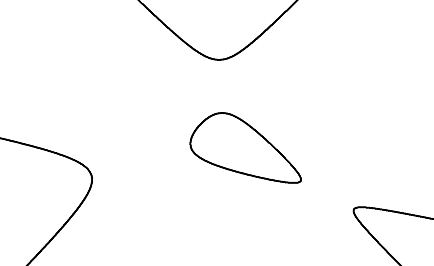}  & \includegraphics[scale = 1]{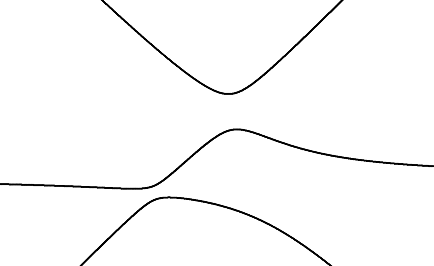} \\ \\
One oval & Two ovals, one bounded & Two unbounded ovals\\
& and one unbounded&
\\
\\
$\mathcal T_C \simeq 4\R^3$&$\mathcal T_C \simeq 4\R^3 \,\sqcup\, 4{S}^1 \times \R^2$&$\mathcal T_C \simeq 4\R^3 \,\sqcup\, 2{S}^1 \times \R^2$
 \end{tabular}

\quad\par
\caption{Types of cubic curves and the topology of corresponding isospectral sets}
\label{cc}
\end{figure}

The answer to the above questions i) and ii) is formulated in terms of the {real part} of the curve \nolinebreak $C$. By definition, the real part $C_\R$ of a real projective curve $C$ is the set of its real points: $C_\R= C \cap \RP^2$. If the curve $C$ is smooth, then its real part consists of a finite number of closed curves, which are called {ovals}. An oval is called {bounded} if it does not intersect the line at infinity. Otherwise, it is {unbounded}. 
It is a classical result that a smooth projective cubic can have either one oval, which is then unbounded, or two ovals, at least one of which is unbounded (see Figure \ref{cc}).
\par
Before we describe the set $\mathcal T_C$, note that for this set to be non-empty, the curve $C$ should have the properties listed in Proposition  \ref{scProp}, i.e. it should be a real cubic passing through the points $\infty_1$, $\infty_2$, and $\infty_3$ where $\infty_i = (1 : a_i :0)$, and $a_1$, $a_2$, $a_3$ are the eigenvalues of the matrix $A$. The following theorem in particular says that for smooth curves $C$ these conditions are also sufficient for the set $\mathcal T_C$ to be non-empty.
 \begin{theorem} \label{surfacetheorem}
Assume that $A$ is a generic diagonal matrix, and let $C$ be a smooth real cubic passing through the points $\infty_1$, $\infty_2$, and $\infty_3$. Then the following statements hold.
\begin{enumerate}
\item If the real part of $C$ has one oval, then the set $\mathcal T_C$ has four connected components each diffeomorphic to $\R^3$.
\item If the real part of $C$ has two ovals and one of them is bounded, then
$\mathcal T_C$ has four components diffeomorphic to $\R^3$ and four components diffeomorphic to ${S}^1 \times \R^2$.
\item Finally, if the real part of $C$ has two unbounded ovals, then
$\mathcal T_C$ has four components diffeomorphic to $\R^3$ and two components diffeomorphic to ${S}^1 \times \R^2$.
\end{enumerate}
Furthermore, all solutions of \eqref{mainEquation} lying on components of $\mathcal T_C$ diffeomorphic to $\R^3$ blow up\footnote{In what follows, when we say that a solution blows up, we mean that it does so both forward and backward in time. Note that equation \eqref{mainEquation}  does have solutions which blow up only in one direction, but these solutions correspond to singular spectral curves.}, while all solutions lying on ${S}^1 \times \R^2$ components exist for all times.
\end{theorem}

Note that this theorem does not answer the third of the above questions. Namely, if the spectral curve $C_X$ has two ovals,  then Theorem \ref{surfacetheorem} does not allow us to determine whether a solution with the initial condition $X$ blows up or exists for all times. As we discuss below, the answer to this question can also be given in terms of a simple geometric construction.

\begin{remark}
Note that despite the fact that solutions located on $\R^3$ components blow up, the topology of these components is still compatible with the Arnold-Liouville theorem. This phenomenon is explained in the first author's paper \cite{aleshkin2014topology}.
\end{remark}
\subsection{Regularly inscribed hexagons}\label{rig}

\par

As was pointed out above, it is not in general possible to decide from the spectral curve~$C_X$ whether a solution of equation \eqref{mainEquation} with the initial condition $X$ blows up. So, we need to supplement the curve with some additional data in order to be able to understand the behavior of a given solution. It turns out that as such additional data we can take a certain hexagon inscribed in the spectral curve. This hexagon is constructed as follows.
\par
As before, we assume that $A$ is a diagonal matrix with distinct diagonal entries $a_1$, $a_2$, $a_3$. Under this assumption, the spectral curve $C_X$ has three distinct real asymptotes which are, by definition, the tangent lines to $C_X$ at the points $\infty_1$, $\infty_2$, $\infty_3$. Denote these asymptotes by $l_1$, $l_2$,  $l_3$. Let $X \in \gl(3,\R)$ be such that the spectral curve $C_X$ is smooth. Consider the matrix
$$
X_{{z}} :=  z_3 X +  z_1 A - z_2 \Id = \left(\begin{array}{ccc}x_{11}z_3 + a_1z_1 - z_2 & x_{12}z_3 & x_{13}z_3 \\x_{21}z_3 & x_{22}z_3 + a_2z_1 - z_2  & x_{23}z_3  \\x_{31}z_3  & x_{32}z_3  & x_{33}z_3 + a_3z_1 - z_2 \end{array}\right).
$$
Recall that the zero locus of the determinant of this matrix is, by definition, the spectral curve \nolinebreak $C_X$. First we observe the following:
\begin{proposition}\label{asym}
Asymptotes of the curve $C_X$ are the zero loci of the diagonal entries of the matrix $X_z$. In other words, the equation of the asymptote $l_i$ is $L_i = 0$ where $$L_i := x_{ii}z_3 + a_iz_1 - z_2. $$
\end{proposition}
\begin{proof}

We have $$\det X_z = L_1 L_2 L_3 + z_3^2 L$$ where $L$ is a linear function in $z_1$, $z_2$, $z_3$. Therefore, the restriction of the function $L_i$ to the curve $\det X_z = 0$ has a zero of order $2$ at infinity, which means that $L_i = 0$ is an asymptote.

\end{proof}
Further, let $M_{ij}(z_1, z_2, z_3)$  be the $(i,j)$ minor of the matrix $X_z$. Then, for $i \neq j$, we have $M_{ij} = z_3 L_{ij}$ where $L_{ij}$ is a linear function  in $z_1$, $z_2$, $z_3$.  Explicitly, one has
\begin{align}\label{LIJF}
L_{ij} = \pm \det \left(\begin{array}{cc} x_{ki} & x_{kk}z_3 + a_kz_1 - z_2   \\x_{ji}  & x_{jk}z_3  \end{array}\right)
\end{align}
where $(i,j,k)$ is any permutation of $(1,2,3)$. Note that the function $L_{ij}$ cannot be identically equal to zero. Indeed, as follows from formula \eqref{LIJF}, if $L_{ij} \equiv 0$, then either the $j$'th row or the $i$'th column of the matrix $ X_z$ contains two zeros. The latter implies that the polynomial $\det X_z$ is reducible, which contradicts the smoothness of the curve $C_X$. \par
This way, we obtain six straight lines $l_{ij}$ given by $L_{ij} = 0$ where $i \neq j$. Properties of these straight lines are described in the following proposition.

\begin{figure}[t]
 \centering

\includegraphics[scale = 1.5]{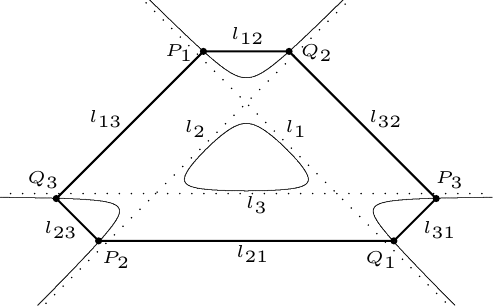}

\caption{Inscribed hexagon $H_X$}
\label{hex}
\end{figure}
\begin{samepage}
\begin{proposition}\label{propsOfLIJ} Let $(i,j,k)$ be any permutation of $(1,2,3)$. Then:
\begin{enumerate}
\item The line $l_{ij}$ is parallel to the asymptote $l_k$ of the curve $C_X$. In other words, we have
$$
l_{12} \parallel l_{21} \parallel l_3, \quad l_{23} \parallel l_{32} \parallel l_1, \quad l_{31} \parallel l_{13} \parallel l_2.
$$
\item We have $l_{ij} \neq l_{ik}$, and $l_{ji} \neq l_{ki}$.
\item The points $$P_i := l_{ij} \cap l_{ik}, \quad Q_i := l_{ji} \cap l_{ki}$$
lie in the real part of the curve $C_X$. In other words, $P_1 Q_2 P_3 Q_1 P_2 Q_3$ is an inscribed hexagon (see Figure \ref{hex}).
\end{enumerate}
\end{proposition}
\end{samepage}
\begin{proof}
The first statement is straightforward and follows from Proposition \ref{asym} and formula~\eqref{LIJF} for the function $L_{ij}$. Let us prove the second statement. Assume that $l_{ij} = l_{ik} = l$. Then, since $l_{ij} \parallel l_k$ and $l_{ik} \parallel l_j$, we have $\infty_j, \infty_k \in l$. Therefore, $l$ is the line at infinity. At the same time, it easy to see from formula \eqref{LIJF} that the line $l_{ij} $ is the line at infinity if and only if $x_{ji} = 0$. So, we have $x_{ji} = x_{ki} = 0$. However, if this was so, then the curve $C_X$ would be not smooth but reducible. Therefore, our assumption is false, and $l_{ij} \neq l_{ik}$. The proof of the inequality $l_{ji} \neq l_{ki}$ is analogous.\par
Now, let us prove the third statement. We shall demonstrate that the point $P_1$ lies in the real part of $C_X$. The proof for other points is analogous. First note that $P_1$ is the intersection point of two real straight lines, therefore this point is real. So, it suffices to show that $P_1 \in C_X$. Let $P_1 = (z_1 : z_2 : z_3)$. Then, by definition of the point $P_1$, the first two columns, as well as the first and the last column of the matrix 
\begin{align}\label{twoByThree}
\left(\begin{array}{ccc} x_{21} & x_{22}z_3 + a_2z_1 - z_2  & x_{23}z_3  \\x_{31}  & x_{32}z_3  & x_{33}z_3 + a_3z_1 - z_2 \end{array}\right)
\end{align}
are linearly dependent. Note that the first column of this matrix cannot be zero: if it is zero, then, again, $C_X$ is a reducible curve. Therefore, the rank of this matrix is equal to one, which implies that $\det X_z = 0$, and thus $P_1 \in C_X$. 
\end{proof}
\begin{remark}\label{quadrics}
Note that since the rank of the matrix \eqref{twoByThree} is equal to $1$ at the point $P_1$, the diagonal minor $M_{11}$ of the matrix $X_z$ at $P_1$ is equal to zero. Similarly, $M_{11}$ vanishes at the point $Q_1$. Also note that the zero locus of the minor $M_{11}$ is a quadric whose asymptotes coincide with the asymptotes $l_2$, $l_3$ of the spectral curve. The latter implies that the quadric $M_{11} = 0$ has at most two finite intersection points with the spectral curve, and these points are $P_1$ and $Q_1$. Similarly, $P_i$ and $Q_i$ may be defined as finite intersection points of the quadric $M_{ii} = 0$ with the spectral curve.
\end{remark}
Proposition \ref{propsOfLIJ} implies that with each matrix $X$ such that the corresponding spectral curve $C_X$ is smooth, one can associate a hexagon inscribed in the real part of the spectral curve. We denote this hexagon by $H_X$. The sides of this hexagon are parallel to the asymptotes of the curve $C_X$. In what follows, hexagons with this property are called \textit{regularly inscribed hexagons}. More precisely, we give the following definition:
\begin{definition}\label{rihDef}
Assume that $C$ is a real smooth cubic curve which intersects the line at infinity at real points $\infty_1$, $\infty_2$, and $ \infty_3$. A hexagon \textit{regularly inscribed} in $C$ is six points  $P_1$, $Q_2$, $P_3$, $Q_1$, $P_2$, $Q_3 \in C_\R$ such that for any permutation $(i,j,k)$ of $(1,2,3)$, the third intersection point of the line $P_iQ_j$ with the curve $C$ is $\infty_k$.
\end{definition}
Note that a regularly inscribed hexagon is uniquely determined by any of its vertices.  Indeed, assume that we are given a point $P_1 \in C_\R$. Then we can reconstruct the point $Q_2$ as the intersection of the curve with the line passing through $P_1$  and parallel to the asymptote $l_3$. In a similar way, we reconstruct points $P_3$, $Q_1$, $P_2$, and $Q_3$. If it turns out that $Q_3P_1$ is parallel to the asymptote $l_2$, then we obtain a regularly inscribed hexagon. In what follows, we show that this is always so. This, in fact, is a simple corollary of a classical result about nine points on a cubic, known as Chasles, or Cayley-Bacharach theorem. Thus, there exists exactly one regularly inscribed hexagon with a given vertex $P_1$ (this is, of course, true for other vertices as well). In particular, the set of hexagons regularly inscribed in a curve $C$ can be, in principle, identified with the real part of $C$.\par 
Also note that Definition \ref{rihDef} describes a slightly more general class of hexagons compared to Proposition \ref{propsOfLIJ}. Indeed, setting $P_i = Q_i = \infty_i$ for $i = 1,2,3$, we obtain a regularly inscribed hexagon. In what follows, we shall refer to this hexagon as \textit{the degenerate hexagon} (note that if $P_i = \infty_i$ or $Q_i = \infty_i$ at least for one value of $i$, then the hexagon is automatically degenerate). Since all sides of the degenerate hexagon coincide, the second statement of Proposition \ref{propsOfLIJ} implies that this hexagon does not correspond to any matrix $X$.  As we show below, this situation is exceptional: any other regularly inscribed hexagon corresponds to a $2$-dimensional family of matrices $X$.\par

 \par

 \begin{figure}[t]
 \centering
\begin{tabular}{cc}
\multicolumn{1}{l}{$(0,6)$} & \multicolumn{1}{l}{$(6,0)$}
\\
{\includegraphics[scale = 1.5]{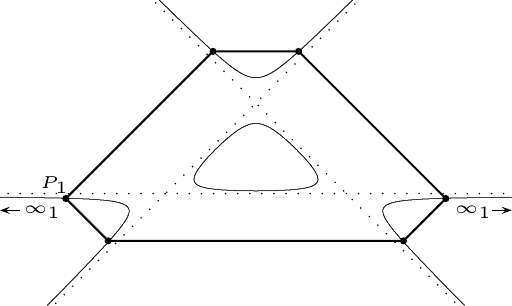}} \quad & \quad {\includegraphics[scale = 1.5]{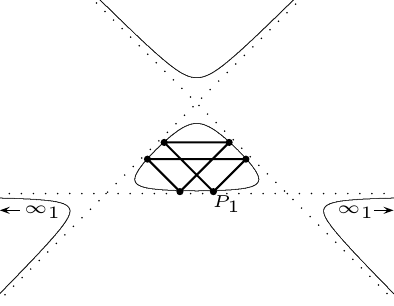}} \\
Blow-up solution & Global solution
\\ \\
 \\
 \multicolumn{1}{l}{$(4,2)$} & \multicolumn{1}{l}{$(2,4)$}
  \\{\includegraphics[scale = 1.5]{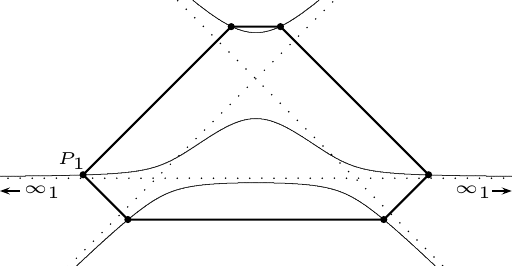}}\quad & \quad{\includegraphics[scale = 1.5]{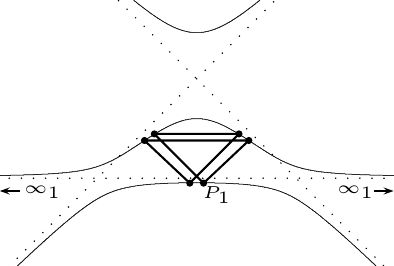}} \\
Blow-up solution & Global solution \\ \\
\end{tabular}
\caption{Different types of regularly inscribed hexagons and the behavior of corresponding solutions}
\label{hex2}
\end{figure}

Now, we need to discuss some topological properties of regularly inscribed hexagons. Note that if the real part of $C$ has two ovals, then different vertices of a regularly inscribed hexagon may lie on different ovals. To distinguish between possible configurations, we use the fact that if the real part of a cubic has two ovals, then exactly one of these two ovals is contractible in $\RP^2$. Namely, if one of the ovals is bounded, then it is contractible, and the other oval is not; if both ovals are unbounded, then the one which intersects the line at infinity at two points is contractible, and the other one is not. 
\begin{definition}
We say that a regularly inscribed hexagon has \textit{type} $(m,n)$ if $m$ of its vertices lie on the contractible oval, and $n$ of its vertices of lie on the non-contractible oval.
\end{definition}
 \begin{proposition}
All possible types of regularly inscribed hexagons are depicted in Figure~\ref{hex2}: if one of the ovals of $C_\R$ is bounded, then $H$ has type $(0,6)$ or $(6,0)$, and if both ovals of $C_\R$ are unbounded, then $H$ has type $(4,2)$ or $(2,4)$.
\end{proposition}
\begin{proof}
The proof follows from simple topological considerations.

\end{proof}

\subsection{Blow-up and global solutions}
Now, we formulate a theorem which allows one to determine whether a given solution of equation \nolinebreak \eqref{mainEquation} blows up. This result is stated in terms of the type of the hexagon $H_X$.\par  First note that, in contrast to the spectral curve $C_X$, the hexagon  $H_X$ is time-dependent. However, the {type} of this hexagon obviously cannot change with time. Moreover, the type of  $H_X$ stays the same if we vary $X$ within a connected component of the isospectral set $\mathcal T_C$. It turns out that  the topological type of  $H_X$ allows us to distinguish between $\R^3$ and ${S}^1 \times \R^2$ components of $\mathcal T_C$, i.e. between global and blow-up solutions.
\begin{theorem} \label{completenesstheorem}
Assume that  $A$ is a generic diagonal matrix, and let $X \in \gl(3,\R)$ be such that the spectral curve $C_X$ is smooth. Let also $X(t)$ be a solution of equation~\eqref{mainEquation}  with the initial condition $X$. Then the following statements hold.
\begin{enumerate}
\item If the real part of $C_X$ has one oval, then $X(t)$ blows up.
\item If the real part of $C_X$  has two ovals, then $X(t)$ blows up if and only if the hexagon \nolinebreak$H_X$ is of type $(0,6)$ or $(4,2)$; if $H_X$ is of type $(6,0)$ or $(2,4)$, then $X(t)$ exists for all times.

\end{enumerate}
\end{theorem}
\begin{example}[Rigid body]\label{rbex} As we mentioned in the Introduction, for skew-symmetric matrices $X$ equation \eqref{mainEquation} becomes the Euler equation governing the motion of a rigid body fixed at the center of mass.  Let us demonstrate how Theorem \ref{completenesstheorem} works in this case.\par
 The equation of the spectral curve, written in affine coordinates $\lambda = z_1 / z_3,$ and $\mu = z_2/ z_3$, is $\det(X + \lambda A - \mu \Id) = 0$. Using that $X^t = -X$, we have
\begin{align*}
 \det(X + \lambda A - \mu \Id) &=  \det((X + \lambda A - \mu \Id)^t) = \\
 &  \det(-X + \lambda A - \mu \Id) =  -\det(X - \lambda A + \mu \Id),
\end{align*}

so, for a skew-symmetric $X$, the spectral curve  $C_X$ is symmetric with respect to the origin. The latter in particular implies that $C_X$ has two ovals both of which are unbounded.
 \begin{figure}[t]
 \centering

  {\includegraphics[scale = 1.5]{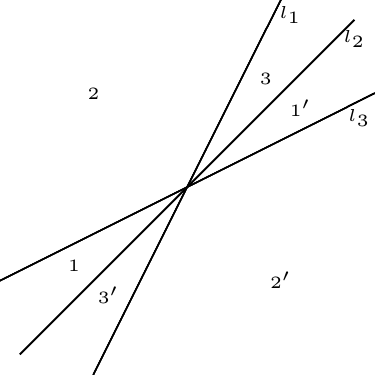}}\caption{Asymptotes of the spectral curve for the rigid body}
\label{asymfig}
\end{figure}

\par
Further, let us show that if the matrix $X$ is skew-symmetric, then the hexagon $H_X$ is of type $(2,4)$.  Figure \ref{asymfig} depicts the asymptotes $l_1$, $l_2$, $l_3$ of the spectral curve and six sectors into which these asymptotes cut the affine plane. The asymptote $l_i$ is given by the equation $\mu = a_i \lambda$ (without loss of generality, we may assume that $a_1 > a_2 > a_3$). According to Remark \ref{quadrics}, the vertex $P_1$ of the hexagon $H_X$ may be found as one of the intersection points of the spectral curve with the quadric $M_{11} = 0$. In affine coordinates, this quadric is given by 
 $$
 (\mu - a_2 \lambda)(\mu - a_3 \lambda) + x_{23}^2 = 0,
 $$
i.e. it is a hyperbola whose branches lie in sectors $\rom{1}$ or $\rom{1}'$. Therefore, the point $P_1$ lies in one of these sectors. The branch of the curve $C_X$ lying in the corresponding sector joins the point $P_1$ with at least one of the points $\infty_2$ or $\infty_3$. So,  $P_1$ lies in the same oval as  $\infty_2$ or $\infty_3$, and thus the hexagon $H_X$ is indeed of type $(2,4)$ (see Figure \ref{hex2}). The latter allows us to conclude that all generic solutions of the rigid body equation exist for all times. This is, of course, very well known (these trajectories are in fact periodic).

\end{example}
\begin{example}[Rigid body in pseudo-Euclidian space] Equation \eqref{mainEquation} may also be restricted to the pseudo-orthogonal Lie algebra $\so(1,2)$, which consists of matrices satisfying the equation $XI + IX^t = 0$
where $I$ is diagonal $I := \mathrm{diag}(1,1,-1)$.
For such matrices $X$, equation \eqref{mainEquation}  may be regarded as the equation of a rigid body in the pseudo-Euclidian space $\R^{1,2}$.\par
As in the Euclidian case, the spectral curve  $C_X$ is symmetric with respect to the origin and thus has two unbounded ovals. The difference is that we can no longer assume that $a_1 > a_2 > a_3$ due to the special role of the first coordinate. 
So, there are two different cases. The first case is when $a_1$ does not lie in the interval $(a_2, a_3)$. Then, repeating the argument of the Euclidian case (see Example \ref{rbex}), one shows that the hexagon $H_X$ is of type $(2,4)$, and thus {all} generic trajectories exist for all times. 
The second case is $a_1 \in (a_2, a_3)$. In this case, a similar argument shows that the hexagon $H_X$ is of type $(4,2)$. Thus, if $a_1$ is between $a_2$ and $a_3$, then all generic trajectories blow up in finite time.
\end{example}

 Now, let us give an informal explanation why Theorem \ref{completenesstheorem} is true. 
First, assume that the real part of the spectral curve $C_X$ has one oval. 
Then, as the matrix $X$ evolves according to equation \eqref{mainEquation}, the hexagon $H_X$, and in particular its vertex $P_1$, slide along the spectral curve $C_X$. At some point of time $t'$, the vertex $P_1$ hits the point $\infty_1$, and the hexagon $H_X$ becomes degenerate. However, as was pointed out above, the degenerate hexagon does not correspond to any matrix $X$. For this reason, the solution $X(t)$ can not be extended to $t = t'$.
\par
 Further, note that if the real part of the spectral curve has two ovals, but the hexagon $H_X$ is of type $(0,6)$ or $(4,2)$, then the points $P_1$ and $\infty_1$ still lie in {the same} oval (see Figure~\ref{hex2}). So, we arrive to exactly the same conclusion as in the one oval case. 
 \par
 Finally, if the real part of $C_X$ has two ovals, and $H_X$ is of type $(6,0)$ or $(2,4)$, then the points $P_1$ and $\infty_1$ lie in different ovals. For this reason, they can never meet each other, and the solution exists for all times.\par
Of course, to turn this explanation into a rigorous proof, one should understand the dynamics of the hexagon $H_X$. It turns out that this dynamics is, roughly speaking, a uniform rotation. More precisely, there exists an angular coordinate $\phi \in [0, 2\pi)$ on each oval of the curve $C_X$ such that the evolution of all vertices of $H_X$ is given by $\diff \phi / \diff t = \mathrm{const} \neq 0$. Thus, if the points $P_1$ and $\infty_1$ lie in the same oval, they are doomed to meet each other. 
\par
Note that this consideration also implies that for each global in time solution of equation \eqref{mainEquation}, the hexagon $H_X$ returns to its initial position after some time $T$. In other words, the evolution of the hexagon $H_X$ is periodic. However, the evolution of the matrix $X$ itself is, in general, not periodic but quasi-periodic. To be more precise, we have the following.
\begin{theorem} \label{phasetheorem}
Assume that $A$ is a generic diagonal matrix, and let $C$ be a smooth real cubic passing through the points $\infty_1$, $\infty_2$, $\infty_3$. Further, assume that the real part of $C$ has two ovals\footnote{Recall that if the real part of $C$ has one oval, then there are no global in time solutions lying in $\mathcal T_C$.}. Then there exist a real number $T > 0$ and a diagonal matrix $M \in \mathbb{P}\mathrm{D}(3, \R)$ such that for each lying in $\mathcal T_C$ global in time solution of equation~\eqref{mainEquation}, the following statements hold.
\begin{enumerate}
\item The dynamics of the hexagon $H_X$ is periodic with period $T$. 
\item The dynamics of the matrix $X$ is quasi-periodic:
$$
X(t + T) = MX(t)M^{-1}.
$$
Furthermore, we have $M \in \mathbb{P}\mathrm{D}^+(3, \R)$ if $H_X$ has type $(6,0)$, and $M \notin \mathbb{P}\mathrm{D}^+(3, \R)$ if $H_X$ has type $(2,4)$.
\end{enumerate}

\end{theorem}
\begin{example}[Rigid body revisited]
Let us again consider the case of a skew-symmetric matrix $X$. Then, as follows from considerations of Example \ref{rbex}, the hexagon $H_X$ has type~$(2,4)$. Therefore, $X(t + T) = MX(t)M^{-1}$ where $M \in \mathbb{P}\mathrm{D}(3, \R) \setminus  \mathbb{P}\mathrm{D}^+(3, \R) $. On the other hand, all generic trajectories of the rigid body are closed, so we should have $M^k = \Id$ for a suitable integer $k>0$. Clearly, this is only possible when the diagonal entries of $M$ are equal to $\pm 1$, and thus $M^2 = \Id$. Note that $M$ itself is not $\pm \Id$ since $M \notin \mathbb{P}\mathrm{D}^+(3, \R)$. So, we have
$
X(t + 2T) = X(t)
$,
 i.e. the period of a generic trajectory of the rigid body is \textit{twice} the period of the corresponding hexagon.
\end{example}
\begin{remark}\label{relGen}
Let us comment on the relation between the hexagon $H_X$ and the general approach of the algebro-geometric integration theory. {In this general approach, one considers the eigenvector of the Lax matrix $X_\lambda$ as a line bundle $E$ over the spectral curve. The fiber of the bundle $E$ at the point $(\lambda, \mu)$ is the eigenspace of $X_\lambda$ corresponding to the eigenvalue $\mu$ (one can show that for smooth spectral curves this eigenspace is always one-dimensional, and that the line bundle $E$ extends to the points at infinity). The isomorphism class of the line bundle $E$ defines a point in the Jacobian of the spectral curve. 
The main result of the algebro-geometric integration theory is that the evolution of this point according to the Lax equation is linear with respect to the addition law on the Jacobian (see, e.g., the above-mentioned monograph \cite{babelon}; cf. Proposition \ref{velocityFormula} below). From the latter it, in particular, follows that equation \eqref{mainEquation} can be solved in terms of theta functions (cf.~\cite{Man}).}\par 
{
The relation between the line bundle $E$ and the hexagon $H_X$ is as follows. For any regularly inscribed hexagon $H$, we have
\begin{align}\label{eqDiv}
P_1 + \infty_2 + \infty_3 \sim \infty_1 + P_2 + \infty_3 \sim \infty_1 + \infty_2 + P_3\,,
\end{align}
where the $D_1 \sim D_2$ denotes the linear equivalence of divisors $D_1$, $D_2$. Indeed, the divisor of the function $L_{13} / L_{23}$, where $L_{ij}$ is given by \eqref{LIJF}, is
$
P_1 + \infty_2 - \infty_1 - P_2
$ (see Figure~\ref{hex}). Therefore, $P_1 + \infty_2 \sim \infty_1 + P_2$, and $P_1 + \infty_2 + \infty_3 \sim \infty_1 + P_2 + \infty_3 $, as desired. The proof of the equivalence $ \infty_1 + P_2 + \infty_3 \sim \infty_1 + \infty_2 + P_3$ is analogous. Furthermore, one can show that the line bundle corresponding to divisors \eqref{eqDiv} is isomorphic to the eigenvector bundle $E$. Thus, the data contained in the hexagon $H_X$ and the line bundle $E$ are equivalent. However, it turns out that it is easier to read off the information about the matrix $X$ from the corresponding hexagon rather than from the line bundle. }

\end{remark}

\section{Proofs of the main results}\label{proofs}
\subsection{Regularly inscribed hexagons and Chasles' theorem}\label{rihCht}
In this section we prove that for any point lying in the real part of a cubic curve, there exists a unique regularly inscribed hexagon whose vertex $P_1$ is at that point. Of course, there is nothing special about the vertex $P_1$, so the result is also true for any other vertex.\par
Let $C$ be a smooth real cubic curve, and assume that $C$ intersects the line at infinity at three real points $\infty_1, \infty_2, \infty_3$. Then $C$ has three real asymptotes $l_1, l_2, l_3$.  
Take any point $P_1 \in C_\R$ and consider the line passing through $P_1$  and parallel to $l_3$, that is the line passing through $P_1$ and \nolinebreak$\infty_3$ (if $P_1 = \infty_3$, then the line passing through $P_1$ and \nolinebreak$\infty_3$ is, by definition, the tangent line to $C$ at $\infty_3$, i.e. the asymptote $l_3$). By B\'{e}zout's theorem, this line should meet the curve $C$ at one more point which we call $Q_2$ (see Figure \ref{hex}). Clearly, $Q_2 \in C_\R$. Now, consider the line passing through $Q_2$ and parallel to \nolinebreak$l_1$, and denote its third intersection point with $C$ by $P_3$. Continuing this procedure, we obtain points $P_1,Q_2,P_3,Q_1,P_2,Q_3,P_1^* \in C_\R$ such that
$$
P_1Q_2 \parallel l_3, \quad Q_2P_3 \parallel l_1,  \quad P_3Q_1 \parallel l_2,  \quad Q_1P_2 \parallel l_3, \quad P_2Q_3 \parallel l_1, \quad Q_3P_1^* \parallel l_2.   
$$
Now, we need to show that the points $P_1^*$ and $P_1$ in fact coincide, so that the polygon $P_1Q_2P_3Q_1P_2Q_3$ is a regularly inscribed hexagon. For simplicity, assume that the nine points $$P_1, Q_2,P_3,Q_1,P_2,Q_3, \infty_1, \infty_2, \mbox{ and } \infty_3$$ are pairwise distinct (the general case follows by continuity). We shall apply the following classical result, known as Chasles', or Cayley-Bacharach theorem:
\begin{theorem}\label{cht}
Let $C_1,C_2,$ and $C_3$ be three cubic curves in $\CP^2$. Assume that $C_1$ intersects $C_2$ at nine distinct points $R_1, \dots, R_9$, and that $C_3$ passes through eight of these nine points. Then $C_3$ also passes through the ninth point.
\end{theorem}
To apply this result in our setting, consider the curves $$C_1 = C,\quad C_2 = P_1Q_2 \cup P_3Q_1 \cup P_2Q_3, \mbox{ and } C_3 = Q_2P_3 \cup Q_1P_2 \cup Q_3P_1^*.$$ 
Then $$C_1 \cap C_2 = \{ P_1,Q_2,P_3,Q_1,P_2,Q_3, \infty_1, \infty_2, \infty_3\},$$ 
and thus $C_3$ passes through all points of $C_1 \cap C_2$ except, possibly, $P_1$. Therefore, by Theorem~\ref{cht}, the curve $C_3$ also passes through the point $P_1$, so $$
P_1 \in C_1 \cap C_3 =  \{ Q_2,P_3,Q_1,P_2,Q_3,P_1^*, \infty_1, \infty_2, \infty_3\},
$$
and hence $P_1 = P_1^*$, q.e.d. \par
\subsection{Reconstructing a matrix from a curve and a hexagon}
In this section we show that a matrix $X$ can be reconstructed from the spectral curve $C_X$ and the hexagon $H_X$ uniquely up to conjugation by diagonal matrices. For the sake of simplicity, we shall assume that the spectral curve satisfies the following genericity assumptions:
i) it does not have inflection points at infinity,
ii) intersection points of asymptotes do not lie on the curve.
It is easy to see that under these assumptions every regularly inscribed hexagon which is not degenerate has at most one side at infinity. The proof of the general case is similar, but one needs to consider more cases.
\begin{proposition}\label{reconstruction} Let $C$ be a smooth real projective cubic passing through the points $\infty_1$, $\infty_2$, $\infty_3$, and satisfying the above genericity assumptions.
Let also $H \subset C$ be a regularly inscribed hexagon which is not degenerate. Then there exists a matrix $X \in \gl(3,\R)$, unique up to conjugation by a diagonal matrix, such that $C_X = X$, and $H_X = H$.
\end{proposition}
\begin{proof}
First note that the spectral curve $C_X$ and the hexagon $H_X$ are invariant under transformations $X \mapsto DXD^{-1}$ where $D$ is a diagonal invertible matrix. Therefore, uniqueness up to conjugation by diagonal matrices is the best result we can expect.\par
Now, let us show how to reconstruct $X$ from $C$ and $H$. Note that by Proposition \ref{asym} the diagonal entries $x_{11}, x_{22}, x_{33}$ of the matrix $X$ are uniquely determined by the spectral curve. So, we only need to reconstruct the off-diagonal terms. \par
 First, assume that none of the sides $H$ are at infinity. This also implies that none of the sides of $H$ are asymptotes of $C$. In terms of the matrix $X$ to be constructed, these conditions mean that all off-diagonal terms of $X$ do not vanish (cf. formula \eqref{LIJF}). Take any two non-zero real numbers $\alpha,\beta \in \R^*$, and set $x_{31} = \alpha, x_{32} = \beta$. Now, we find the remaining entries of $X$ by using the equations of sides of $H$. First, we find $x_{12}$ from the equation of $l_{23}$. Since the side $l_{23}$ is parallel to the asymptote $l_1$, its equation has the form
$$
a(x_{11}z_3 + a_1z_1 - z_2 ) - bz_3 = 0
$$
where $a,b$ are constants. Note that since $l_{23}$ is neither the line at infinity, nor an asymptote, these constants do not vanish. On the other hand, in terms of the matrix $X$ to be constructed, the equation of $l_{23}$ should be
$$
x_{32}(x_{11}z_3 + a_1z_1 - z_2 )  - x_{31}x_{12}z_3  = 0,
$$
so 
$$
\frac{x_{32}}{a} = \frac{x_{31}x_{12}}{b},
$$
which allows us to find $x_{12} $. In a similar way, we find $x_{21}$ from the equation of $l_{13}$, then $x_{23}$ from $l_{12}$, and, finally, $x_{13}$ from $l_{32}$.\par
Now, let us show that the curves $C_X$ and $C$ are the same. As follows from the construction of $X$, the hexagons $H_X$ and $H$ have four common consecutive sides $l_{23}$, $l_{13}$, $l_{12}$, and $l_{32}$. Therefore, they have three common vertices $Q_3, P_1$, and $Q_2$. By Proposition~\ref{propsOfLIJ}, these points are not collinear. So, the curves $C$ and $C_X$ are two cubics which have common asymptotes and three common non-collinear points. As it is easy to see, such cubics must coincide. This, in turn, implies that the hexagons $H$ and $H_X$ coincide as well: they are two hexagons which are regularly inscribed in the same cubic and have a common vertex~$P_1$.\par
Now, assume that $H$ has one side at infinity. Without loss of generality, it is side $l_{31}$. Then two adjacent sides of $H$ are necessarily asymptotes of the curve, namely $l_{32} = l_1$, and $l_{21} = l_3$. Similarly to the first case, we fix $\alpha,\beta \in \R^*$ and set $x_{31} = \alpha, x_{32} = \beta$. Apart from this, since the side $l_{31}$ is at infinity, we should set $x_{13} = 0$  (see the proof of Proposition~\ref{propsOfLIJ}). Further, similarly to the above, we find $x_{12}$, $x_{21}$, $x_{23}$ using the equations of $l_{23}$, $l_{13}$, and $l_{12}$ respectively. Note that since $x_{13} = 0$, the sides $l_{32}$ and $l_{21}$ of the hexagon $H_X$ automatically coincide with the corresponding sides of $H$, that is with asymptotes $l_1, l_3$. So, $H_X$ and $H$ have five sides in common and thus, similarly to the above, $C_X$ coincides with $C$, and $H_X$ coincides with $H$.\par
So, we showed that in both cases a matrix $X$ satisfying the desired conditions is uniquely determined by its non-zero entries $x_{31}, x_{32}$. Conjugating such a matrix with a suitable diagonal matrix, we can always assume that $x_{31} = 1$ and $x_{32} = 1$. Therefore, $X$ is indeed unique up to conjugation by a diagonal matrix.

\end{proof}
\subsection{Description of isospectral sets}\label{topology}
In this section, we prove the topological part of Theorem \ref{surfacetheorem}. Namely, we describe the topology of the sets $\mathcal T_C = \{X \in \gl(3,\R) : C_X = C\}$.\par
Let $\mathcal H_C$ be the set of hexagons regularly inscribed in the curve $C$. Let also $H_d \in \mathcal H_C$ be the degenerate hexagon. Proposition \ref{reconstruction} allows us to conclude that the set $\mathcal T_C$ is the total space of a principal  $\mathbb P\mathrm D(3,\R)$ bundle over $\mathcal H_C \setminus H_d$. Also note that, according to Section \ref{rihCht}, the space $\mathcal H_C \setminus H_d$ can be identified with the real part of the curve $C$ without a point.\par

A trivializing cover for the bundle $\mathcal T_C \to \mathcal H_C \setminus H_d$ can be constructed as follows. Let
$$
U_1 = \{ H \in \mathcal H_C : l_{13} \neq l_2 \}, \quad U_2 = \{ H \in \mathcal H_C : l_{31} \neq l_2 \}.
$$

It it easy to see that, under the genericity assumptions of the previous section, 
$\{U_1, U_2\}$ is indeed a cover of $\mathcal H_C \setminus H_d$.  
As coordinates on fibers over $U_1$, we take the entries $x_{21}$ and $x_{32}$ of the matrix $X$. Since $l_{13} \neq l_2$, these entries are non-zero (cf. formula \eqref{LIJF}), and thus they uniquely determine a matrix $X$ within its $\mathbb P\mathrm D(3,\R)$ orbit. Similarly, we take  $x_{12}$ and $x_{23}$ as coordinates on fibers over $U_2$.\par
Now, let us prove the first statement of Theorem \ref{surfacetheorem}. If the real part $C_\R$ of the curve $C$ has one oval, then the set $\mathcal H_C \simeq C_\R$ is diffeomorphic to a circle ${S}^1$. Therefore, the set $\mathcal H_C \setminus H_d $ is diffeomorphic to $\R$, and hence $\mathcal T_C \to \mathcal H_C \setminus H_d$ is a trivial bundle: $\mathcal T_C \simeq \R \times \mathbb P\mathrm D(3,\R).$ The latter set has four connected components diffeomorphic to $\R^3$, q.e.d.\par
Further, let us prove the second statement. Assume that the real part of $C$ has two ovals one of which is bounded. In this case, the set $\mathcal H_C$ has two connected components distinguished by the type of a hexagon. Let
$$
 \mathcal H_C^{i,j} := \{ H \in \mathcal H_C : H \mbox{ is of type } (i,j) \}, \quad  \mathcal T_C^{i,j} :=  \{ X \in  \mathcal T_C: H_X \in \mathcal H_C^{i,j} \}.
 $$
Then
 $$\mathcal H_C = \mathcal H_C^{6,0} \sqcup \mathcal H_C^{0,6}, \mbox{ and } \mathcal T_C = \mathcal T_C^{6,0} \sqcup \mathcal T_C^{0,6}.$$ Note that since the degenerate hexagon $H_d$ has type $(0,6)$, we have  $$\mathcal H_C \setminus H_d = \mathcal H_C^{6,0} \sqcup (\mathcal H_C^{0,6} \setminus H_d).$$
Therefore, the set $\mathcal T_C^{0,6}$ is the total space of a principal $\mathbb P\mathrm D(3,\R)$ bundle over $\mathcal H_C^{0,6} \setminus H_d$.
 Since the latter set $\mathcal H_C^{0,6} \setminus H_d$  is diffeomorphic to $\R$, this bundle is trivial, and thus $\mathcal T_C^{0,6} \simeq 4\R^3.$ \par Now, let us study the component $\mathcal T_C^{6,0}$. By definition of the set $\mathcal H_C^{6,0}$, all vertices of any hexagon $H_X \in \mathcal H_C^{6,0}$ lie on the bounded oval, and therefore none of the vertices are at infinity. In particular, the side $l_{13}$ of $H_X$ cannot coincide with the asymptote $l_2$. So, the component $ \mathcal H_C^{6,0}$ is completely covered by the chart $U_1$, and thus the bundle $\mathcal T_C^{6,0} \to \mathcal H_C^{6,0}$ is also trivial. Since $\mathcal H_C^{6,0} \simeq {S}^1$, we have $\mathcal T_C^{6,0} \simeq 4{S}^1 \times \R^2$, q.e.d.\par
Finally, let us prove the third statement of Theorem \ref{surfacetheorem}. Assume that the real part of $C$ has two unbounded ovals. Then $$\mathcal H_C \setminus H_d = \mathcal H_C^{2,4} \sqcup (\mathcal H_C^{4,2} \setminus H_d).$$
Similarly to the above, we have $\mathcal H_C^{4,2} \setminus H_d \simeq \R$, therefore the bundle $\mathcal T_C^{4,2} \to \mathcal H_C^{4,2} \setminus H_d $ is trivial, and  $\mathcal T_C^{4,2} \simeq 4\R^3$. 
\par
Now, let us show that the bundle $\mathcal T_C^{2,4} \to \mathcal H_C^{2,4}$ is \textit{not} trivial. Note that there are only two different principal $\mathbb P\mathrm D(3,\R)$ bundles over the circle ${S}^1$, namely the trivial one with the total space $ 4{S}^1 \times \R^2$, and the non-trivial one with the total space  $ 2{S}^1 \times \R^2$. So, as soon as we prove that the bundle is non-trivial, the topology of the total space is uniquely determined.
 \par
  \begin{figure}[t]
 \centering
\includegraphics[scale = 1.5]{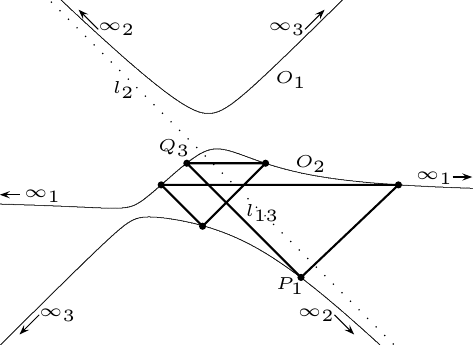} 
\quad\par
\caption{}
\label{hexasymptote}
\end{figure}
 Without loss of generality, we may assume that the points $\infty_2, \infty_3$ lie on the contractible oval $O_1$, and that the point $\infty_1$ lies on the non-contractible oval $O_2$ (if this is not so, we renumber these points). Then $P_1 \in O_1$, and $ Q_3 \in O_2$, as depicted in Figure \nolinebreak\ref{hexasymptote}. Let us consider the intersection of the chart $U_1$ with the component  $\mathcal H_C^{2,4}$. The complement to $U_1$ consists of those hexagons for which $l_{13} = l_2$. The latter is possible if either $P_1 = \infty_2$, or  $Q_3 = \infty_2$. However, for $H_X \in \mathcal H_C^{2,4}$, the points $Q_3$ and  $\infty_2$ lie on different ovals (see Figure \nolinebreak\ref{hexasymptote}), therefore the chart $U_1$ covers the whole set $\mathcal H_C^{2,4} $ except one hexagon $H_0$ distinguished by the condition $P_1 = \infty_2$. This hexagon can be obtained by moving the point $Q_3$ in Figure \ref{hexasymptote} to the right till it reaches the asymptote $l_2$. The domain $\mathcal H_C^{2,4} \setminus H_0$ is covered by the chart $U_1$, therefore in this domain we have a trivialization of the bundle given by $x_{21}, x_{32}$. Let us analyze what happens to these coordinates as the hexagon $H_X$ passes through $H_0$, or, which is the same, as the vertex $Q_3$ crosses the asymptote $l_2$. First, note that the side $l_{12}$ of the hexagon $H_0$ is at infinity, therefore for this hexagon we have $x_{21} = 0$ (see the proof of Proposition \ref{propsOfLIJ}), which shows that the trivialization $\{x_{21}, x_{32}\}$ is no longer valid for $H_X = H_0$. However, since $H_0$ has only one side at infinity, we have $x_{32} \neq 0$. Therefore, to determine whether the bundle  $\mathcal T_C^{2,4} \to \mathcal H_C^{2,4}$ is trivial, we should study what happens to the sign of $x_{21}$ as the point $Q_3$ crosses the asymptote $l_2$. By formula~\eqref{LIJF}, the side $l_{13}$ is given by the equation
 $$
x_{31}(x_{22}z_3 + a_2z_1 - z_2 )  - x_{21}x_{32}z_3  = 0.
$$
At the same time, $x_{22}z_3 + a_2z_1 - z_2 = 0$ is the equation of the asymptote $l_2$. Therefore, provided that $x_{31} \neq 0$, i.e. that the side $l_{13}$ is not at infinity, the sign of the product $x_{21}x_{32}$ has the following geometric meaning: it is positive if the line $l_{13}$ lies on one side of the asymptote $l_2$, and negative if it lies at the other side. Now, notice that as $Q_3$ crosses the asymptote $l_2$, the line $l_{13}$ gets from one side of the asymptote to the other (see Figure \ref{hexasymptote}), therefore the sign of the product $x_{21}x_{32}$ changes to its negative. Since $x_{32}$ does not vanish as  $Q_3$ crosses the asymptote, this means that the sign of $x_{21}$ changes, and therefore the bundle is non-trivial. So, we have $\mathcal T_C^{2,4} \simeq 2{S}^1 \times \R^2$, thus the third statement of Theorem~\ref{surfacetheorem} is proved.
 \subsection{Dynamics of the hexagon}\label{hexDyn}
Let us fix the spectral curve $C$ and describe the evolution of the hexagon $H_X$ under the flow \eqref{mainEquation}.
As before, consider the matrix $ X_z =  z_3 X +  z_1 A - z_2 \Id. $ Recall that the zero locus of the determinant of this matrix is the spectral curve $C$. We shall describe the dynamics of the hexagon $H_X$ by relating its vertices to eigenvectors of $X_z$. Note that since a regularly inscribed hexagon is uniquely determined by any of its vertices, it suffices to describe the dynamics of one vertex. As before, let $M_{ij}$ be the $(i,j)$ minor of $X_z$. Then the vector $(M_{11}, -M_{12}, M_{13})^\mathrm{t}$ belongs to the kernel of $X_z$. Normalizing this vector by dividing its components by $M_{11}$, we obtain the following meromorphic vector-function $\psi$ on $C$:
\begin{align*}
\psi := 
\left(\!\!
\begin{aligned}
&\,\,\,1\\
-&\dfrac{M_{12}}{M_{11}}\\
&\dfrac{M_{13}}{M_{11}}
\end{aligned}
\right)
\end{align*}
Denote the components of this vector by $\psi_1, \psi_2, \psi_3$. By construction, we have $X_z \psi = 0$ identically on $C$. 
\begin{proposition}\label{hProps}
The component $\psi_2$ of $\psi$ has poles at $\infty_2$ and $Q_1$ and zeros at $\infty_1$ and $Q_2$; the component $\psi_3$ has poles at $\infty_3$ and $Q_1$ and zeros at $\infty_1$ and $Q_3$.
\end{proposition}
\begin{proof}
As follows from Remark \ref{quadrics}, zeros of $M_{11}$ on $C$ are the points $P_1, Q_1$, and the points $\infty_2$ and $\infty_3$ taken with multiplicity $2$. Further, we have $M_{12} = z_3 L_{12}$ where $L_{12}$ is a linear function whose zero locus is the line $l_{12}$. Therefore, zeros of $M_{12}$ are the points $P_1$, $Q_2$, $\infty_1$, $\infty_2$, and the point $\infty_3$ taken with multiplicity $2$. Dividing $M_{12}$ by $M_{11}$, we obtain the desired statement about zeros and poles of $\psi_2$. Zeros and poles of $\psi_3$ are computed analogously.
\end{proof}
So, the vertex $Q_1$ of the hexagon $H_X$ is the only movable pole of the eigenvector $\psi$ (clearly, vertices $Q_2$ and $Q_3$ can be interpreted in the same way; to obtain these vertices as poles, one needs to renormalize the vector $\psi$ by setting $\psi_2 = 1$ or $\psi_3 = 1$). This allows us to describe the dynamics of $Q_1$ using standard technique (see e.g. the reviews \cite{DKN, DMN, reyman1994group}). Let $\omega$ be a holomorphic $1$-form on $C$; such a form is unique up to a constant factor. In the affine chart  $\lambda = z_1 / z_3$ and $\mu = z_2 / z_3$, it is given by
\begin{align}\label{hf}
\omega = \frac{\diff \mu}{\partial_\lambda f_X } =  -\frac{\diff \lambda}{\partial_\mu f_X } 
\end{align}
where $f_X(\lambda, \mu)= 0$ is the equation of the affine part of the curve $C$.
\begin{proposition}\label{velocityFormula}
Assume that $X$ evolves according to equation \eqref{mainEquation}. Then, for the holomorphic form $\omega$ normalized by \eqref{hf}, we have
$$
\omega\left( \diffFXp{Q_1}{t}\right) = 1.
$$
\end{proposition}
\begin{proof}
We work in affine coordinates $\lambda = z_1 / z_3$ and $\mu = z_2 / z_3$. Equation $X_z\psi = 0$ can be rewritten as
\begin{align}\label{eve}
(X_\lambda - \mu \Id) \psi = 0
\end{align}
where $X_\lambda = X + \lambda A$. Assuming that $X_\lambda$ evolves according to Lax equation \eqref{LaxFlow} and differentiating \eqref{eve} with respect to time, we get
\begin{align}\label{fee}
(X_\lambda - \mu \Id) \left( \diffFXp{\psi}{t} + Y_\lambda \psi \right) = 0.
\end{align}
Note that  $\rank (X_\lambda - \mu \Id) = \rank X_z = 2$ at every point of $C$. Indeed, if $ \rank X_z = 1$ at some point $P \in C$, then this point is a common zero for all $2 \times 2$ minors $M_{ij}$ of the matrix $X_z$, which means that all vertices of the hexagon $H_X$ coincide. However, as follows from Proposition \ref{propsOfLIJ}, this is impossible. Thus, \eqref{fee} implies that
\begin{align}\label{evde}
 \diffFXp{\psi}{t} = ( \xi\cdot \Id - Y_\lambda) \psi\,,
\end{align}
where $\xi$ is a function on the curve $C$. Using that $\psi_1 \equiv 1$, we find that $\xi$ is equal to the first coordinate $(Y_\lambda \psi)_1$ of the vector $Y_\lambda \psi$.\par
Now, let us consider the point $Q_1$ at some moment of time $t_0$, and let $u$ be a local coordinate on $C$ near this point. Let also $u_q(t)$ be the $u$-coordinate of $Q_1$ at moment $t$. Provided that $Q_1$ does not coincide with $ \infty_2$ or $\infty_3$, we have
\begin{align}\label{pe}
\psi(t) = \frac{h(t)}{u - u_q(t)} + \mbox{terms holomorphic in } u\,,
\end{align}
where $h(t)$ is a vector holomorphic in $u$. Substituting \eqref{pe} into \eqref{evde} and equating coefficients in $(u-u_q(t))^{-2}$, we get 
\begin{align*} \diffFXp{u_q}{t} = (Y_\lambda h)_1\,,
\end{align*}
where $(Y_\lambda h)_1$ is the first coordinate of the vector $Y_\lambda h$. The coordinate free-form of this equation is
\begin{align}\label{pde}
\omega\left( \diffFXp{Q_1}{t}\right) = \Res_{Q_1}((Y_\lambda \psi)_1 \omega)\,,
\end{align}
where $\omega$ is the holomorphic form defined above. Now, note that  the form $(Y_\lambda \psi)_1 \omega$ may only have poles at those points where either $Y_\lambda$ or $\psi$ have a pole, i.e. at points $Q_1$, $\infty_1$, $\infty_2$, $\infty_3$. Therefore, by Cauchy's residue theorem, we have
\begin{align}
\label{residues}
\Res_{Q_1}((Y_\lambda \psi)_1 \omega) = -\sum_{i=1}^3  \Res_{\infty_i}((Y_\lambda \psi)_1 \omega)\,.
\end{align}
Further, note that
$
Y_\lambda = \lambda^{-1}(X_\lambda^2 - X^2)
$ (cf. formula \eqref{xy}), therefore,
\begin{align*}
\begin{aligned}
\Res_{\infty_i}((Y_\lambda \psi)_1 \omega) = \Res_{\infty_i}((X_\lambda^2 \psi)_1 \lambda^{-1}\omega) -\Res_{\infty_i}((X^2 \psi)_1\lambda^{-1} \omega) = \Res_{\infty_i}(\mu^2 \lambda^{-1}\omega) 
\end{aligned}
\end{align*}
where we used the identities $X_\lambda \psi = \mu \psi$, $\psi_1 = 1$, and that the function $(X^2 \psi)_1\lambda^{-1}$ does not have a pole at the point $\infty_i$. Combining the last formula with \eqref{pde} and \eqref{residues}, we conclude that
\begin{align}
\label{velocity}
\omega\left( \diffFXp{Q_1}{t}\right) = -\sum_{i=1}^3 \Res_{\infty_i}(\mu^2 \lambda^{-1}\omega).
\end{align}
Note that although we assumed in the proof that $Q_1 \neq \infty_2$ and $Q_1 \neq \infty_3$, formula~\eqref{velocity} still holds for these points by continuity argument. Now, to complete the proof, it suffices to compute the residues. This can be easily done using formula \eqref{hf} and the explicit expression for the polynomial $f_X$ (note that only cubic terms of $f_X$ affect the residues).
\end{proof}
Note that the conclusion of Proposition \ref{velocityFormula} is obviously true for other $Q$-vertices of the hexagon $H_X$ as well, i.e.
$$
\omega\left( \diffFXp{Q_1}{t}\right) = \omega\left( \diffFXp{Q_2}{t}\right) = \omega\left( \diffFXp{Q_3}{t}\right) = 1\,.
$$
For $P$-vertices, we have
$$
\omega\left( \diffFXp{P_1}{t}\right) = \omega\left( \diffFXp{P_2}{t}\right) = \omega\left( \diffFXp{P_3}{t}\right) = -1\,.
$$
The latter can be proved by noting that \eqref{mainEquation} is anti-invariant with respect to transformation $X \mapsto X^\mathrm{t}$ which preserves the spectral curve $C$, and interchanges $P$-vertices with $Q$-vertices.\par
Now, for each oval $O_i$ of the curve $C$, fix a point $R_i \in O_i$ and consider the function
$$
\phi(R) = \int_{R_i}^{R} \omega\,.
$$
Then $\phi$ is a periodic coordinate on $O_i$. In terms of the coordinate $\phi$, the dynamics of vertices of $H_X$ is linear: $\diff \phi / \diff t = \pm 1$. This in particular implies that the dynamics of the hexagon $H_X$ is monotonous and periodic. The period is given by the integral of $\omega$ along any of the ovals of $C$ (note that this integral is the same for both ovals since they are homologous cycles in $C$).
\subsection{Complete and blow-up solutions}
In this section, we prove Theorems \ref{completenesstheorem} and \ref{phasetheorem}, i.e. we investigate the dynamics of \eqref{mainEquation} at each connected component of the set $\mathcal T_C$. 
As we know from Section \ref{topology}, the set $\mathcal T_C$ is the total space of a principal $\mathbb P\mathrm D(3,\R)$ bundle over the set $\mathcal H_C \setminus H_d$ where $\mathcal H_C$ is the set of all hexagons regularly inscribed in $C$, and $H_d$ is the degenerate hexagon. Furthermore, flow~\eqref{mainEquation} is invariant with respect to the $\mathbb P\mathrm D(3,\R)$ action on $\mathcal T_C$. This allows us to apply the following classical result.
\begin{theorem}[A.\,Lichnerowicz \cite{lichnerowicz1976global}]\label{lt}
Let $\pi \colon E \to B$ be a principal $G$-bundle, and let $v$ be a vector field on $E$ which is invariant with respect to the $G$-action. Then an integral trajectory $x(t)$ of the field $v$ is complete if and only if the corresponding trajectory of the field $\pi_*v$ on the base is complete.
\end{theorem}
This result implies that one can study the behavior of solutions of \eqref{mainEquation} by considering their projections to the space of hexagons. Dynamics of the space of hexagons was studied in the previous section: it is linear dynamics.\par
To prove Theorem \ref{completenesstheorem}, we consider separately each of the components $\mathcal T_C^{i,j}$ defined in Section~\ref{topology}.
For example, consider a solution $X(t)$  of equation \eqref{mainEquation} such that the corresponding hexagon $H_X$ has type $(0,6)$. The latter means that $X(t) \in \mathcal T_C^{0,6}$ in the notation of Section~\ref{topology}. The set $\mathcal T_C^{0,6}$ is a principal bundle over $\mathcal H_C^{0,6} \setminus H_d$. Furthermore, as follows from Section \ref{hexDyn}, the dynamics on $\mathcal H_C^{0,6}$ is linear in terms of the coordinate $\phi$, so the projection of $X(t)$ to the base $\mathcal H_C^{0,6} \setminus H_d$ meets the degrease hexagon $H_d$ and thus blows up in finite time. Therefore, by Theorem~\ref{lt}, the trajectory $X(t)$ itself also blows up.\par
An analogous consideration shows that trajectories of \eqref{mainEquation}  corresponding to hexagons of type $(4,2)$ also blow up, while trajectories corresponding to types $(6,0)$ or $(2,4)$ exist for all times. Thus, Theorem \ref{completenesstheorem} is proved.
\par
Note that this consideration also proves the dynamical part of Theorem \ref{surfacetheorem} since all $\R^3$ components of $\mathcal T_C$ correspond to hexagons of type $(0,6)$ or $(4,2)$, while all ${S}^1 \times \R^2$ components correspond to $(6,0)$ or $(2,4)$ (see Section \ref{topology}).
\par
Now, let us prove Theorem \ref{phasetheorem}. As follows from the previous section, if a trajectory $X(t)$ exists for all times, then the dynamics of the corresponding hexagon is periodic with some period $T$. Now, the formula
$$
X(t + T) = MX(t)M^{-1}.
$$
easily follows from the $\mathbb P\mathrm D(3,\R)$ invariance of the flow \eqref{mainEquation}. Further, let us show that $M \in \mathbb{P}\mathrm{D}^+(3, \R)$ if $H_X$ has type $(6,0)$. Consider the bundle $\mathcal T_C^{6,0} \to \mathcal H_C^{6,0}$. By definition of the number $T$, the matrices $X(t+T)$ and $X(t)$ lie in the same fiber of this bundle. Furthermore, since it is a trivial bundle, $X(t+T)$ and $X(t)$ should lie in the same connected component of the fiber, and thus $M \in \mathbb{P}\mathrm{D}^+(3, \R)$. \par Analogously, since the bundle $\mathcal T_C^{2,4} \to \mathcal H_C^{2,4}$ is not trivial, we have $M \notin \mathbb{P}\mathrm{D}^+(3, \R)$ if $H_X$ has type $(2,4)$, so Theorem \ref{phasetheorem} is proved.

\section{Discussion}\label{discussion}
Note that equation \eqref{mainEquation}, as well as the definitions of the spectral curve and the corresponding set $\mathcal T_C$, can be without any difficulty generalized to $\gl_n(\R)$. It is an interesting question how to generalize our results to this case. In particular, is it always true that the topology of $\mathcal T_C$ is completely determined by the geometry of $C_\R$? It is particularly interesting whether the structure of the set $\mathcal T_C$ depends on the way the ovals of $C$ are nested into each other. Note that description of all possible relative positions of ovals for a real algebraic curve of degree $n$ is the Hilbert 16th problem which is still unsolved in full generality.
\par
We also note that the problem of description of the set $\mathcal T_C$ can be reformulated in purely algebro-geometric terms. Similarly to the $\gl(3,\R)$ case, let $a_1, \dots, a_n$ be the eigenvalues of the matrix $A$. For simplicity, let all eigenvalues $a_1, \dots, a_n$ be real, so that we may assume that $A$ is a diagonal matrix. Further, let $\infty_i = (1 : a_i : 0) \in \CP^2$, and let $C$ be a smooth real projective curve of degree $n$ passing through the points $\infty_1, \dots, \infty_n$. Then, as follows from the results of~\cite{beauville2}, the manifold $\mathcal T_C = \{ X \in \gl(n, \R) : C_X = C\}$ can be described as a principal $(\R^*)^{n-1}$ bundle over the real part of $\Jac(C) \setminus \Theta$ where $\Jac(C)$ is the Jacobian of the curve $C$, and $\Theta \subset \Jac(C)$ is the theta-divisor (the latter can also be deduced from the earlier papers \cite{mvm, Adler, Reyman2}).
\par
Further, it is a classical result by Comessatti (see e.g. \cite{seppala1989moduli}) that the topology of the real part of $\Jac(C)$ is uniquely determined by the number of ovals of $C$. However, one still needs to understand the structure of the \nolinebreak $(\R^*)^{n-1}$ bundle $ \mathcal T_C \to (\Jac(C) \setminus \Theta)_\R.$
 For $n=3$ this is (in much more elementary terms) done in the present paper.
\par
 Also note that there is a more explicit algebro-geometric description of the set $ \mathcal{T}_C$ due to L.\,Gavrilov \cite{Gavrilov}. Namely, consider the singular curve $C_s$ obtained from the smooth curve $C$ by identifying the points $\infty_1, \dots, \infty_n$, and let $\Jac(C_s)$ be the generalized Jacobian of $C_s$. Then, as follows from \cite{Gavrilov}, {the set $\mathcal T_C$ is diffeomorphic to the real part of $\Jac(C_s) \setminus \pi^{-1}(\Theta)$} where $\pi \colon \Jac(C_s) \to \Jac(C)$ is the canonical projection. 
Thus, to describe the set $\mathcal T_C$, one needs to describe the real part of the generalized Jacobian $\Jac(C_s)$, and to study how it intersects the preimage of the theta divisor under the projection $\pi$.

 \printbibliography

\end{document}